\documentclass[runningheads,a4paper]{llncs}
\pdfoutput=1

\usepackage{amssymb}
\usepackage{graphicx}
\usepackage{epstopdf}

\usepackage{url}
\urldef{\mailsa}\path|{fabien.mathieu,|
\urldef{\mailsc}\path|ludovic.noirie}@alcatel-lucent.com|    
\newcommand{\keywords}[1]{\par\addvspace\baselineskip
\noindent\keywordname\enspace\ignorespaces#1}

\usepackage{amsmath,amssymb}

\usepackage{accents} 
\usepackage[colorlinks,citecolor=blue,urlcolor=blue]{hyperref}
\newcommand{\tq}{\text{ s.t. }}
\usepackage{subfigure}
\usepackage{tikz}
\usetikzlibrary{intersections}
\usetikzlibrary{fadings}
\usetikzlibrary{positioning}
\DeclareMathOperator{\linearspan}{vect}

\newcommand{\moinssim}{\leq}

\newcommand{\R}{\mathbb{R}}
\newcommand{\N}{\mathbb{N}}
\newcommand{\intent}[2]{\ensuremath{\{ #1, \ldots, #2 \}}} 
\newcommand{\ens}[1]{\mathcal{#1}}
\newcommand{\bipoint}[1]{\overrightarrow{#1}}
\newcommand{\vecteur}[1]{\overrightarrow{#1}}
\newcommand{\vectutil}[1]{\accentset{\sim}{#1}}
\newcommand{\vectutilapprob}[1]{\accentset{\sim}{#1}}
\newcommand{\classelin}[1]{\accentset{\sim}{#1}}

\newcommand{\bra}[1]{\langle \, #1 \, |} 
\newcommand{\bracket}[2]{\langle \, #1 \, | \, #2 \, \rangle} 
\newcommand{\bigbracket}[2]{\left\langle\left. \, #1 \, \right| \, #2 \, \right\rangle}

\newtheorem{thm}{Theorem}

\newtheorem{prop}[thm]{Proposition}
\newtheorem{lem}[thm]{Lemma}

\newtheorem{exemple}[thm]{Example}

\newcommand{\savedx}{0}
\newcommand{\savedy}{0}
\newcommand{\savedz}{0}

\newcommand{\rotateRPY}[4][0/0/0]
{   \pgfmathsetmacro{\rollangle}{#2}
    \pgfmathsetmacro{\pitchangle}{#3}
    \pgfmathsetmacro{\yawangle}{#4}

    \pgfmathsetmacro{\newxx}{cos(\yawangle)*cos(\pitchangle)}
    \pgfmathsetmacro{\newxy}{sin(\yawangle)*cos(\pitchangle)}
    \pgfmathsetmacro{\newxz}{-sin(\pitchangle)}
    \path (\newxx,\newxy,\newxz);
    \pgfgetlastxy{\nxx}{\nxy};

    \pgfmathsetmacro{\newyx}{cos(\yawangle)*sin(\pitchangle)*sin(\rollangle)-sin(\yawangle)*cos(\rollangle)}
    \pgfmathsetmacro{\newyy}{sin(\yawangle)*sin(\pitchangle)*sin(\rollangle)+ cos(\yawangle)*cos(\rollangle)}
    \pgfmathsetmacro{\newyz}{cos(\pitchangle)*sin(\rollangle)}
    \path (\newyx,\newyy,\newyz);
    \pgfgetlastxy{\nyx}{\nyy};

    \pgfmathsetmacro{\newzx}{cos(\yawangle)*sin(\pitchangle)*cos(\rollangle)+ sin(\yawangle)*sin(\rollangle)}
    \pgfmathsetmacro{\newzy}{sin(\yawangle)*sin(\pitchangle)*cos(\rollangle)-cos(\yawangle)*sin(\rollangle)}
    \pgfmathsetmacro{\newzz}{cos(\pitchangle)*cos(\rollangle)}
    \path (\newzx,\newzy,\newzz);
    \pgfgetlastxy{\nzx}{\nzy};

    \foreach \x/\y/\z in {#1}
    {   \pgfmathsetmacro{\transformedx}{\x*\newxx+\y*\newyx+\z*\newzx}
        \pgfmathsetmacro{\transformedy}{\x*\newxy+\y*\newyy+\z*\newzy}
        \pgfmathsetmacro{\transformedz}{\x*\newxz+\y*\newyz+\z*\newzz}
        \xdef\savedx{\transformedx}
        \xdef\savedy{\transformedy}
        \xdef\savedz{\transformedz}     
    }
}

\tikzset{RPY/.style={x={(\nxx,\nxy)},y={(\nyx,\nyy)},z={(\nzx,\nzy)}}}

\begin{document}

\mainmatter  

\title{Geometry on the Utility Space}

\titlerunning{Geometry on the Utility Space}

\author{Fran\c cois Durand\inst{1}\and Beno\^it Kloeckner\inst{2}\and Fabien Mathieu\inst{3}\and Ludovic Noirie\inst{3}}

\institute{Inria, Paris, France, \email{francois.durand@inria.fr}%
\and Universit\'e Paris-Est, Laboratoire d'Analyse et de Math\'ematiques Appliqu\'ees (UMR 8050), UPEM, UPEC, CNRS, F-94010, Cr\'eteil, France, \email{benoit.kloeckner@u-pec.fr}%
\and Alcatel-Lucent Bell Labs France, Nozay, France, 
\mailsa\mailsc}

\toctitle{Geometry on the Utility Space}
\tocauthor{F. Durand, B. Kloeckner, F. Mathieu and L. Noirie}
\maketitle

\begin{abstract}

We study the geometrical properties of the utility space (the space of expected utilities over a finite set of options), which is commonly used to model the preferences of an agent in a situation of uncertainty.
We focus on the case where the model is neutral with respect to the available options, i.e. treats them, \emph{a priori}, as being symmetrical from one another.
Specifically, we prove that the only Riemannian metric that respects the geometrical properties and the natural symmetries of the utility space is the round metric.
This canonical metric allows to define a uniform probability over the utility space and to naturally generalize the Impartial Culture to a model with expected utilities.
\end{abstract}

\section{Introduction}

\noindent\textbf{Motivation.}
In the traditional literature of Arrovian social choice \cite{arrow1950difficulty}, especially voting theory, the preferences of an agent are often represented by ordinal information only\footnote{This is not always the case: for example, Gibbard~\cite{gibbard1978lotteries} considers voters with expected utilities over the candidates.}: a strict total order over the available options, or sometimes a binary relation of preference that may not be a strict total order (for example if indifference is allowed). However, it can be interesting for voting systems to consider cardinal preferences, for at least two reasons.

Firstly, some voting systems are not based on ordinal information only, like Approval voting or Range voting.

Secondly, voters can be in a situation of uncertainty, either because the rule of the voting system involves a part of randomness, or because each voter has incomplete information about other voters' preferences and the ballots they will choose.
To express preferences in a situation of uncertainty, a classical and elegant model is the one of expected utilities \cite{vonneumann1944games,fishburn1970utility,kreps1990microeconomic,mascolell1995microeconomic}. The utility vector $\vecteur{u}$ representing the preferences of an agent is an element of $\R^m$, where $m$ is the number of available options or \emph{candidates}; the utility of a lottery over the options is computed as an expected value.

For a broad set of applications in economics, the options under consideration are financial rewards or quantities of one or several economic goods,
which leads to an important consequence: there is a natural structure over the space of options.
For example, if options are financial rewards, there is a natural order over the possible amounts and it is defined prior to the agents' preferences.%

We consider here the opposite scenario where options are symmetrical \emph{a priori}. This symmetry condition is especially relevant in voting theory, by a normative principle of neutrality, but it can apply to other contexts of choice under uncertainty when there is no natural preexistent structure on the space of available options.

The motivation for this paper came from the possible generalizations of the \emph{Impartial Culture} to agents with expected utilities. The Impartial Culture is a classical probabilistic model in voting theory where each agent draws independently her strict total order of preference with a uniform probability over the set of possible orders.

The difficulty is not to define a probability law over utilities such that its projection on ordinal information is the Impartial Culture. Indeed, it is sufficient to define a distribution where voters are independent and all candidates are treated symmetrically. The issue is to choose one in particular: an infinity of distributions meet these conditions and we can wonder whether one of these generalizations has canonical reasons to be chosen rather than the others.

To answer this question, we need to address an important technical point. As we will see, an agent's utility vector is defined up to two constants, and choosing a specific normalization is arbitrary. As a consequence, the utility space is a quotient space of $\R^m$, and \emph{a priori}, there is no canonical way to push a metric from $\R^m$ to this quotient space. Hence, at first sight, it seems that there is no natural definition of a uniform distribution of probability over that space.

More generally, searching a natural generalization of the Impartial Culture to the utility space naturally leads to investigate different topics about the geometry of this quotient space and to get a better understanding of its properties related to algebra, topology and measure theory.

We emphasize that our goal is not to propose a measure that represents real-life preferences. Such approach would follow from observation and experimental studies rather than theoretical work. Instead, we aim at identifying a measure that can play the role of a neutral, reference, measure. This will give a model for uniformness to which real distributions of utilities can be compared. For example, if an observed distribution has higher density than the reference measure in certain regions of the utility space, this could be interpreted as a non-uniform culture for the population under study and give an indication of an overall trend for these regions. With this aim in mind, we will assume a symmetry hypothesis over the different candidates.

\noindent\textbf{Contributions and plan.}
The rest of the paper is organized as follows. In Section~\ref{sec:modele_VNM}, we quickly present Von Neumann--Morgenstern framework and define the utility space. In Section~\ref{sec:dualite}, we show that the utility space may be seen as a quotient of the dual of the space of pairs of lotteries over the candidates. In Section~\ref{sec:operateurs}, we naturally define an inversion operation, that corresponds to reversing preferences while keeping their intensities, and a summation operation, that is characterized by the fact that it preserves unanimous preferences.

Since the utility space is a manifold, it is a natural wish to endow it with a Riemannian metric. In Section~\ref{sec:riemann}, we prove that the only Riemannian representation that preserves the natural projective properties and the \emph{a priori} symmetry between the candidates is the round metric. Finally, in Section~\ref{sec:pratique}, we use this result to give a canonical generalization of the Impartial culture and to suggest the use of Von Mises--Fisher model to represent polarized cultures.

\vspace{-.4cm}
\section{Von Neumann--Morgenstern Model}\label{sec:modele_VNM} %
\vspace{-.2cm}

In this section, we define some notations and we quickly recall Von Neumann--Morgenstern framework in order to define the utility space.

Let $m \in \N \setminus \{0\}$. We will consider $m$ mutually exclusive options called \emph{candidates}, each one represented by an index in $\intent{1}{m}$.
A \emph{lottery} over the candidates is an $m$-tuple $(L_1,\ldots,L_m) \in (\R_{+})^m$ such that (s.t.) $\sum_{j=1}^{m}{L_j} = 1$. The set of lotteries is denoted $\ens{L}_m$.

The preferences of an agent over lotteries are represented by a binary relation $\moinssim$ over $\ens{L}_m$.

Von Neumann--Morgenstern theorem states that, provided that relation $\moinssim$ meets quite natural assumptions\footnote{The necessary and sufficient condition is that relation $\moinssim$ is complete, transitive, archimedean and independent of irrelevant alternatives.}, it can be represented by a utility vector $\vecteur{u} = (u_1,\dots,u_m) \in \R^m$, in the sense that following the relation $\moinssim$ is equivalent to maximizing the expected utility. Formally, for any two lotteries $L$ and $M$:
\vspace{-.15cm}
$$
L \moinssim M \Leftrightarrow \sum_{j=1}^{m}{L_j u_j} \leq \sum_{j=1}^{m}{M_j u_j}\text{.}
$$
\vspace{-.15cm}

Mathematical definitions, assumptions and proof of this theorem can be found in \cite{vonneumann2007games,mascolell1995microeconomic,kreps1990microeconomic}, and discussions about the experimental validity of the assumptions are available in \cite{fishburn1988nonlinear,mascolell1995microeconomic}.

For the purposes of this paper, a crucial point is that $\vecteur{u}$ is defined up to an additive constant and a positive multiplicative constant. Formally, for any two vectors $\vecteur{u}$ and $\vecteur{v}$, let us note $\vecteur{u} \sim \vecteur{v}$ iff $\exists a \in \mathopen(0,+\infty\mathclose), \exists b \in \R$ s.t. $\vecteur{v} = a \vecteur{u} + b \vecteur{1}$, where $\vecteur{1}$ denotes the vector whose $m$ coordinates are equal to 1. With this notation, if $\vecteur{u} \in \R^m$ is a utility vector representing $\moinssim$, then a vector $\vecteur{v} \in \R^m$ is also a utility vector representing $\moinssim$ iff $\vecteur{u} \sim \vecteur{v}$.%

In order to define the utility space, all vectors representing the same preferences must be identified as only one point.
The \emph{utility space} over $m$ candidates, denoted $\ens{U}_m$, is defined as the quotient space $\R^m / \sim$.
We call \emph{canonical projection} from $\R^m$ to $\ens{U}_m$ the function:
\vspace{-.2cm}
$$
\vectutil{\pi}:
\begin{array}{|cll}
\R^m & \to & \ens{U}_m \\
\vecteur{u} & \to & \vectutil{u} = \{\vecteur{v} \in \R^m \tq \vecteur{v} \sim \vecteur{u}\}.
\end{array}
$$
\vspace{-.15cm}
For any $\vecteur{u} \in \R^m$, we denote without ambiguity $\moinssim_{\vectutil{u}}$ the binary relation over $\ens{L}_m$ represented by $\vecteur{u}$
.

\begin{figure}
\begin{center}
\begin{tikzpicture}[scale=1.3,z={(-0.6,-0.3)}]
\fill[lightgray,path fading=east, fading angle=0] (1.3,1.3,1.3) -- ++(1,-0.3,0) -- (-1.3+1,-1.3-0.3,-1.3) -- (-1.3,-1.3,-1.3);
\draw[>=latex,->] (0,0,0)-- (1,1,1) node[left,yshift=-.1cm,xshift=0.03cm] {$\vecteur{1}$};
\draw[>=latex,->] (0,0,0)--(1.3+1,1.3-0.3,1.3) node[right,black,inner sep=2pt] {$\vecteur{u}^{(2)}$};
\draw[>=latex,->,scale=0.7] (0,0,0)--(1.3+1,1.3-0.3,1.3) node[below right,black,inner sep=1pt] {$\vecteur{u}^{(3)}$};
\draw[>=latex,->,scale=0.7] (0,0,0)--(1.3+1+0.55,1.3-0.3+0.55,1.3+0.55) node[above,black,inner sep=2pt] {$\vecteur{u}^{(1)}$};
\draw[>=latex,->] (0,0,0)--(-1.3+1,-1.3-0.3,-1.3) node[right,black,inner sep=3pt] {$\vecteur{u}^{(4)}$};
\draw (-1.2,0,0)--(1.2,0,0);
\draw (0,-1.2,0)--(0,1.2,0);
\draw (0,0,0)--(0,0,1.3);
\draw[>=latex,->]  (0,0,0)--(0,0,1) node[above left,inner sep=1pt] {$\vecteur{e_1}$};
\draw[>=latex,->]  (0,0,0)--(1,0,0) node[below] {$\vecteur{e_2}$};
\draw[>=latex,->]  (0,0,0)--(0,1,0) node[left] {$\vecteur{e_3}$};
\node at (0.9,-0.7) {$\vectutil{u}$};
\end{tikzpicture}
\end{center}
\vspace{-0.5cm}\caption{Space $\R^3$ of utility vectors for 3~candidates (without quotient).}\label{fig:halfplane}
\vspace{-.4cm}
\end{figure}
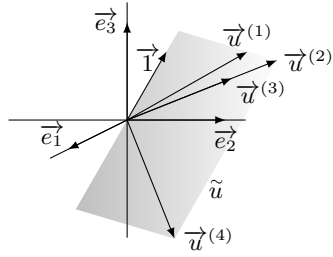

Figure~\ref{fig:halfplane} represents the space $\R^3$ of utility vectors for 3~candidates, without projecting to the quotient space. The canonical base of $\R^3$ is denoted $(\vecteur{e_1}, \vecteur{e_2}, \vecteur{e_3})$. 
Utility vectors $\vecteur{u}^{(1)}$ to $\vecteur{u}^{(4)}$ represent the same preferences as any other vector of the half-plane $\vectutil{u}$, represented in gray. More generally, each non-trivial point $\vectutil{u}$ in the quotient utility space corresponds to a half-plane in $\R^m$, delimited by the line $\linearspan(\vecteur{1})$, the linear span of $\vecteur{1}$.
The only exception is the point of total indifference $\vectutil{0}$. In $\R^m$, it does not correspond to a plane but to the line $\linearspan(\vecteur{1})$ itself.

\vspace{-.7cm}

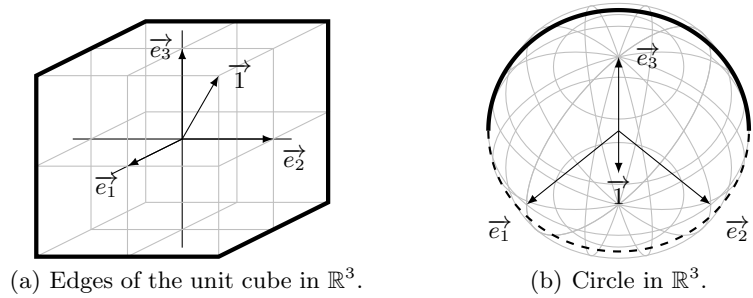
\begin{figure}
\begin{center}
\subfigure[Edges of the unit cube in $\R^3$.]{%
~~~\begin{tikzpicture}[scale=1.2,z={(-0.6,-0.3)}]
\draw (-1.2,0,0)--(1.2,0,0);
\draw (0,-1.2,0)--(0,1.2,0);
\draw (0,0,0)--(0,0,1.3);
\draw[gray!50] (1,1,1) -- (-1,1,1);
\draw[gray!50] (1,1,1) -- (1,-1,1);
\draw[gray!50] (1,1,1) -- (1,1,-1);
\draw[gray!50] (-1,-1,-1) -- (-1,-1,1);
\draw[gray!50] (-1,-1,-1) -- (1,-1,-1);
\draw[gray!50] (-1,-1,-1) -- (-1,1,-1);
\draw[gray!50] (-1,1,0) -- (1,1,0) -- (1,-1,0) -- (-1,-1,0) -- cycle;
\draw[gray!50] (-1,0,1) -- (1,0,1) -- (1,0,-1) -- (-1,0,-1) -- cycle;
\draw[gray!50] (0,-1,1) -- (0,1,1) -- (0,1,-1) -- (0,-1,-1) -- cycle;
\draw[>=latex,->] (0,0,0)-- (1,1,1) node[right] {$\vecteur{1}$};
\draw[>=latex,->]  (0,0,0)--(1,0,0) node[below right] {$\vecteur{e_2}$};
\draw[>=latex,->]  (0,0,0)--(0,1,0) node[left] {$\vecteur{e_3}$};
\draw[>=latex,->]  (0,0,0)--(0,0,1) node[below left] {$\vecteur{e_1}$};
\draw[ultra thick] (1,-1,-1) -- (1,1,-1) -- (-1,1,-1) -- (-1,1,1) -- (-1,-1,1) -- (1,-1,1) -- cycle;
\end{tikzpicture}~~~
\label{fig:cubic}%
}
\hspace{1.2cm}
\subfigure[Circle in $\R^3$.]{%
\tikzset{
	MyPersp/.style={scale=1,x={(-0.866cm,-0.7cm)},y={(0.866cm,-0.7cm)},
    z={(0cm,0.7cm)}},
	MyPoints/.style={fill=white,draw=black,thick}
		}
\begin{tikzpicture}[MyPersp]
\begin{scope}[scale=1.4]
	\foreach \t in {0,30,...,180}
		{\draw[gray!50] ({cos(\t)},{sin(\t)},0)
			\foreach \rho in {5,10,...,360}
				{--({cos(\t)*cos(\rho)},{sin(\t)*cos(\rho)},
      {sin(\rho)})}--cycle;
		}
	\foreach \t in {-90,-60,...,90}
		{\draw[gray!50] ({cos(\t)},0,{sin(\t)})
			\foreach \rho in {5,10,...,360}
				{--({cos(\t)*cos(\rho)},{cos(\t)*sin(\rho)},
      {sin(\t)})}--cycle;
		}
	\foreach \t in {5,10,...,360}
		{--({cos(\t)},{sin(\t)},0)}--cycle;
\draw[>=latex,->]  (0,0,0)--(1,0,0) node[below left=3pt] {$\vecteur{e_1}$};
\draw[>=latex,->]  (0,0,0)--(0,1,0) node[below right=3pt] {$\vecteur{e_2}$};
\draw[>=latex,->]  (0,0,0)--(0,0,1) node[right=3pt] {$\vecteur{e_3}$};
\end{scope}

\rotateRPY{0}{-35.26}{45}
\begin{scope}[RPY,scale=1.4]
\draw[->,>=latex] (0,0,0) -- (1,0,0) node[below=-1pt] {$\vecteur{1}$};
\draw[ultra thick] [domain=0:pi,  samples=80,  smooth] plot (0, {cos(\x r)}, {sin(\x r)})  ;
\draw[thick,dashed] [domain=pi:2*pi,  samples=80,  smooth] plot (0, {cos(\x r)}, {sin(\x r)})  ;
\end{scope}
\end{tikzpicture}
\label{fig:circle}%
}
\caption{Two representations of $\ens{U}_3$.}
\end{center}
\vspace{-.7cm}
\end{figure}

To deal with utilities, conceptually and practically, it would be convenient to have a canonical representative $\vecteur{u}$ for each equivalence class $\vectutil{u}$. In Figure~\ref{fig:cubic}, for each non-indifferent $\vectutil{u}$, we choose its representative satisfying $\min(u_i) = -1$ and $\max(u_i) = 1$. The utility space $\ens{U}_3$ (except the indifference point) is represented in $\R^3$ by six edges of the unit cube.
In Figure~\ref{fig:circle}, we choose the representative satisfying $\sum u_i = 0$ and $\sum u_i^2 = 1$. In that case, $\ens{U}_3 \setminus \{\vectutil{0}\}$ is represented in $\R^3$ by the unit circle in the linear plane that is orthogonal to $\vecteur{1}$.

If we choose such a representation, the quotient space $\ens{U}_m$ can inherit the Euclidean distance from $\R^m$; for example, we can evaluate distances along the edges of the cube in Figure~\ref{fig:cubic}, or along the unit circle in Figure~\ref{fig:circle}. But it is clear that the result will depend on the representation chosen. Hence, it is an interesting question whether one of these two representations, or yet another one, has canonical reasons to be used. But before answering this question, we need to explore in more generality the geometrical properties of the utility space.

\section{Duality with the Tangent Hyperplane of Lotteries}\label{sec:dualite}

In this section, we remark that the utility space is a dual of the space of pairs of lotteries. Not only does it give a different point of view on the utility space (which we think is interesting by itself), but it will also be helpful to prove Theorem~\ref{thm:unanimite}, which characterizes the summation operator that we will define in Section~\ref{sec:operateurs}.

In the example represented in Figure~\ref{fig:duality}, we consider $m=3$ candidates and $\vecteur{u} = (\frac{5}{3}\,;-\frac{1}{3}\,;-\frac{4}{3})$.
The great triangle, or simplex, is the space of lotteries $\ens{L}_3$.
Hatchings are the agent's indifference lines: she is indifferent between any pair of lotteries on the same line (see \cite{mascolell1995microeconomic}, Section 6.B).
The utility vector $\vecteur{u}$ represented here is in the plane of the simplex,
but it is not mandatory: indeed, $\vecteur{u}$ can be arbitrarily chosen in its equivalence class $\vectutil{u}$. Nevertheless, it is a quite natural choice, because the component of $\vecteur{u}$ in the direction $\vecteur{1}$ (orthogonal to the simplex) has no meaning in terms of preferences.

\begin{figure}
\begin{center}
\begin{tikzpicture}[scale=2]
\draw[>=latex,->]  (0,0,0)--(1.2,0,0) node[right] {$L_1$};
\draw[>=latex,->]  (0,0,0)--(0,1.2,0) node[above] {$L_2$};
\draw[>=latex,->]  (0,0,0)--(0,0,1.2) node[left] {$L_3$};
\draw (1,0,0) node[above right]{1} --(0,1,0) node[above right]{1} --(0,0,1) node[below right]{1}--cycle;
\node[above right,inner sep=1pt] at (0.4,0.6,0) {$\ens{L}_3$};
\foreach  \c  in  {-1,-0.8,...,-0.2}  {
\draw (0,1+\c,-\c)--({(1+\c)/3},0,{(2-\c)/3});
}
\foreach  \c  in  {0,0.2,...,2}  {
\draw (\c/2,1-\c/2,0)--({(1+\c)/3},0,{(2-\c)/3});
}
\foreach  \c  in  {0.4}  {
\draw[>=latex,->] ({(2+5*\c)/12},{(2-\c)/4},{(2-\c)/6}) -- +(5/15,-1/15,-4/15) node[above right] {$\vecteur{u}$};
\draw ({(2+5*\c)/12},{(2-\c)/4},{(2-\c)/6}) ++(5/150,-1/150,-4/150) -- ++(3/150/1.73,-9/150/1.73,6/150/1.73) -- ++(-5/150,+1/150,+4/150);
}
\end{tikzpicture}
\end{center}
\vspace{-0.5cm}\caption{Space $\mathcal{L}_3$ of lotteries over 3~candidates.}\label{fig:duality}
\end{figure}
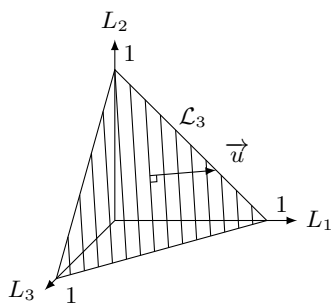

The utility vector $\vecteur{u}$ can be seen as a gradient of preference: at each point, it reveals in what directions the agent can find lotteries she likes better.
However, only the direction of $\vecteur{u}$ is  important, whereas its norm has no specific meaning. As a consequence, the utility space is not exactly a dual space, but rather a quotient of a dual space, as we will see more formally.

For any two lotteries $L = (L_1,\ldots,L_m)$ and $M = (M_1,\ldots,M_m)$, we call \emph{bipoint} from $L$ to $M$ the vector 
$
\bipoint{LM} = (M_1-L_1,\ldots,M_m-L_m).
$
We call \emph{tangent polytope} of $\ens{L}_m$ the set $\ens{T}$ of bipoints of $\ens{L}_m$.

We call \emph{tangent hyperplane} of $\ens{L}_m$:
$$
\ens{H} = \{(\Delta_1,\ldots,\Delta_m) \in \R^m \tq \sum_{j=1}^{m} \Delta_j = 0\}.
$$

In Figure~\ref{fig:duality}, the tangent polytope $\mathcal{T}$ is the set of the bipoints of the great triangle, seen as a part of a vector space (whereas $\ens{L}_m$ is a part of an affine space).
The tangent hyperplane $\mathcal{H}$ is the whole linear hyperplane containing $\mathcal{T}$.

Let us note $\bracket{\vecteur{u}}{\vecteur{v}}$ the canonical inner product of $\vecteur{u}$ and $\vecteur{v}$.
We call \emph{positive half-hyperplane associated to $\vecteur{u}$} the set
$
\vecteur{u}^{+} = \{\vecteur{\Delta} \in \ens{H} \tq \bracket{\vecteur{u}}{\vecteur{\Delta}} \geq 0\}.
$
By definition, a lottery $M$ is preferred to a lottery $L$ iff the bipoint $\bipoint{LM}$ belongs to this positive half-hyperplane:
$$
L \moinssim_{\vectutil{u}} M \Leftrightarrow \bracket{\vecteur{u}}{\bipoint{LM}} \geq 0 \Leftrightarrow \bipoint{LM} \in \vecteur{u}^{+}.
$$

Let $\ens{H}^\star$ be the dual space of $\ens{H}$, that is, the set of linear forms on~$\ens{H}$.
For any $\vecteur{u} \in \R^m$, we call \emph{linear form associated to $\vecteur{u}$} the following element of $\ens{H}^\star$:
$$
\bra{\vecteur{u}}:
\begin{array}{|cll}
\mathcal{H} & \to & \R \\
\vecteur{\Delta} & \to & \bracket{\vecteur{u}}{\vecteur{\Delta}}.
\end{array}
$$
We observed that the utility vector can be seen as a gradient, except that only its direction matters, not its norm. Let us formalize this idea.
For any $(f,g) \in (\ens{H}^\star)^2$, we denote $f \sim g$ iff these two linear forms are positive multiples of each other, that is, iff $\exists a \in \mathopen(0,+\infty\mathclose)$ s.t. $g = a f$. We denote $\classelin{\pi}(f) = \{g \in \ens{H}^\star \text{ s.t. } f \sim g\}$: it is the equivalence class of $f$, up to positive multiplication.

\begin{prop}
For any $(\vecteur{u},\vecteur{v}) \in (\R^m)^2$, we have:
$$
\vecteur{u} \sim \vecteur{v} \Leftrightarrow \bra{\vecteur{u}} \sim \bra{\vecteur{v}}.
$$

The following application is a bijection:
$$
\Theta:
\begin{array}{|cll}
\ens{U}_m & \to & \ens{H}^\star / \sim \\
\vectutil{\pi}(\vecteur{u}) & \to & \classelin{\pi}(\bra{\vecteur{u}}). 
\end{array}
$$
\end{prop}

\begin{proof}
$\vecteur{u} \sim \vecteur{v}$

$\Leftrightarrow \exists a \in \mathopen(0,+\infty\mathclose), \exists b \in \R \tq 
\vecteur{v} - a \vecteur{u} = b \vecteur{1}$

$\Leftrightarrow \exists a \in \mathopen(0,+\infty\mathclose) \tq 
\vecteur{v} - a \vecteur{u}$ is orthogonal to $\ens{H}$

$\Leftrightarrow \exists a \in \mathopen(0,+\infty\mathclose) \tq 
\forall \vecteur{\Delta} \in \ens{H}, \bracket{\vecteur{v}}{\vecteur{\Delta}} = a \bracket{\vecteur{u}}{\vecteur{\Delta}}$

$\Leftrightarrow \bra{\vecteur{u}} \sim \bra{\vecteur{v}}$.

The implication $\Rightarrow$ proves that $\Theta$ is correctly defined: indeed, if $\vectutil{\pi}(\vecteur{u}) = \vectutil{\pi}(\vecteur{v})$, then $\classelin{\pi}(\bra{\vecteur{u}}) = \classelin{\pi}(\bra{\vecteur{v}})$.
The implication $\Leftarrow$ ensures that $\Theta$ is injective. Finally, $\Theta$ is obviously surjective.
\end{proof}

Hence, the utility space can be seen as a quotient of the dual~$\mathcal{H}^\star$ of the tangent space~$\mathcal{H}$ of the lotteries $\ens{L}_m$.
As noticed before, a utility vector may be seen, up to a positive constant, as a uniform gradient, i.e. as a linear form over $\mathcal{H}$ that reveals, for any point in the space of lotteries, in what directions the agent can find lotteries that she likes better.

\section{Inversion and Summation Operators}\label{sec:operateurs}

As a quotient of $\R^m$, the utility space inherits natural operations from $\R^m$: inversion and summation.
We will see that both these quotient operators have an intuitive meaning regarding preferences.
The summation will also allow to define \emph{lines} in Section~\ref{sec:riemann}, which will be a key notion for Theorem~\ref{thm_sphere}, characterizing the suitable Riemannian metrics.

We define the \emph{inversion} operator of $\ens{U}_m$ as:
$$
-: 
\begin{array}{|cll}
\ens{U}_m & \to & \ens{U}_m \\
\vectutil{\pi}(\vecteur{u}) & \to & \vectutil{\pi}(-\vecteur{u}).
\end{array}
$$
This operator is correctly defined and it is a bijection; indeed, $\vectutil{\pi}(\vecteur{u}) = \vectutil{\pi}(\vecteur{v})$ iff $\vectutil{\pi}(-\vecteur{u}) = \vectutil{\pi}(-\vecteur{v})$.
Considering this additive inverse amounts to reverting the agent's preferences, without modifying their relative intensities.

Now, we want to push the summation operator from $\R^m$ to the quotient $\mathcal{U}_m$.
We use a generic method to push an operator to a quotient space: considering $\vectutil{u}$ and $\vectutil{v}$ in $\mathcal{U}_m$, their antecedents are taken in $\R^m$ thanks to $\vectutil{\pi}^{-1}$, the sum is computed in $\R^m$, then the result is converted back into the quotient space $\mathcal{U}_m$, thanks to $\vectutil{\pi}$.

However, the result is not unique.
Indeed, let us take arbitrary representatives $\vecteur{u} \in \vectutil{u}$ and $\vecteur{v} \in \vectutil{v}$.
In order to compute the sum, we can think of any representatives. So, possible sums are $a \vecteur{u} + a' \vecteur{v} + (b + b') \vecteur{1}$, where $a$ and $a'$ are positive and where $b + b'$ is any real number.
Converting back to the quotient, we can get for example $\vectutil{\pi}(2 \vecteur{u} + \vecteur{v})$ and $\vectutil{\pi}(\vecteur{u} + 3 \vecteur{v})$, which are generally distinct.
As a consequence, the output is not a point in the utility space $\mathcal{U}_m$, but rather a set of points, i.e. an element of $\mathcal{P}(\ens{U}_m)$.

This example shows how we could define the sum of two elements $ \vectutil{u} $ and $ \vectutil{v} $. In order to be more general, we will define the sum of any number of elements of $\mathcal{U}_m$. Hence we also take $\mathcal{P}(\ens{U}_m)$ as the set of inputs.

We define the \emph{summation} operator as:
$$
\sum:
\begin{array}{|cll}
\mathcal{P}(\ens{U}_m) & \to & \mathcal{P}(\ens{U}_m) \\
A & \to & \left\{ \vectutil{\pi}\left(\sum_{i=1}^{n} \vecteur{u_i}\right), n \in \N,
(\vecteur{u_1},\ldots,\vecteur{u_n}) \in \left(\vectutil{\pi}^{-1}(A)\right)^n
 \right\}.
\end{array}
$$

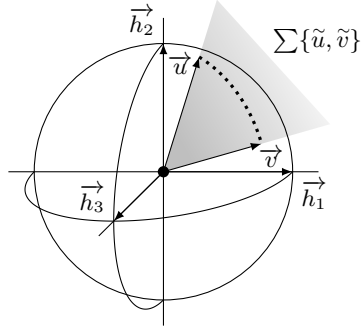
\begin{figure}
\begin{center}
\begin{tikzpicture}[scale=1.7]
\foreach  \theta  in  {20}  {
\fill[lightgray,path fading=south, fading angle=120] ({1.5*cos(\theta) * sin(20)}, {1.5*cos(20)}, {1.5*sin(\theta) * sin(20)}) -- (0,0,0)--({1.7*cos(\theta) * sin(70)}, {1.7*cos(70)}, {1.7*sin(\theta) * sin(70)});
\draw[>=latex,->] (0,0,0)--({cos(\theta) * sin(20)}, {cos(20)}, {sin(\theta) * sin(20)}) node[left,yshift=-.1cm,xshift=0.03cm] {$\vecteur{u}$};
\draw[>=latex,->] (0,0,0)--({cos(\theta) * sin(70)}, {cos(70)}, {sin(\theta) * sin(70)}) node[right,yshift=-.18cm,xshift=-0.15cm] {$\vecteur{v}$};
\draw [variable=\phi,domain=20:70,  samples=50,  smooth, very thick, dotted] plot  ({cos(\theta) * sin(\phi)}, {cos(\phi)}, {sin(\theta) * sin(\phi)},) ;
\node[right] at ({1.5*cos(\theta) * sin(40)}, {1.5*cos(40)}, {1.5*sin(\theta) * sin(40)},) {$\sum \{\vectutil{u},\vectutil{v}\}$};
}
\draw (-1.2,0,0)--(1.2,0,0);
\draw (0,-1.2,0)--(0,1.2,0);
\draw (0,0,0)--(0,0,1.3);
\draw[>=latex,->]  (0,0,0)--(1,0,0) node[below right] {$\vecteur{h_1}$};
\draw[>=latex,->]  (0,0,0)--(0,1,0) node[above left] {$\vecteur{h_2}$};
\draw[>=latex,->]  (0,0,0)--(0,0,1) node[above left] {$\vecteur{h_3}$};
\draw (0,0,0) circle (1);
\draw [domain=0:pi,  samples=80,  smooth] plot  (0, {cos(\x r)}, {sin(\x r)})  ;
\draw [domain=-pi/2:pi/2,  samples=80,  smooth] plot  ({sin(\x r)}, 0, {cos(\x r)})  ;
\node[fill=black,shape=circle,inner sep=1.5pt] at (0,0,0) {};
\end{tikzpicture}
\end{center}
\vspace{-0.5cm}\caption{Sum of two utility vectors in the utility space $\ens{U}_4$.}\label{fig:sum_vectors}
\end{figure}

\begin{exemple}
Let us consider $\mathcal{U}_4$, the utility space for 4 candidates. In Figure~\ref{fig:sum_vectors}, for the purpose of visualization, we represent its projection in $\mathcal{H}$, which is permitted by the choice of normalization constants $b$. Since $\mathcal{H}$ is a 3-dimensional space, let $(\vecteur{h_1},\vecteur{h_2},\vecteur{h_3})$ be an orthonormal base of it.

For two non-trivial utility vectors $\vectutil{u}$ and $\vectutil{v}$, the choice of normalization multiplicators~$a$ allow to choose representatives $\vecteur{u}$ and $\vecteur{v}$ whose Euclidean norm is~1.

In this representation, the sum $\sum \{\vectutil{u},\vectutil{v}\}$ consists of utilities corresponding to vectors $a \vecteur{u} + a' \vecteur{v}$, where $a$ and $a'$ are nonnegative. Indeed, we took a representation in $\mathcal{H}$, so all normalization constants $b$ vanish. Moreover, $a$, $a'$ or both can be equal to zero because our definition allows to ignore $\vecteur{u}$, $\vecteur{v}$ or both.
Up to taking representatives of unitary norm for non-trivial utility vectors, let us note that the sum $\sum \{\vectutilapprob{u},\vectutilapprob{v}\}$ can be represented by the dotted line and the point $\vecteur{0}$ of total indifference.
\end{exemple}

Geometrically, the sum is the quotient of the convex hull of the inputs. Note that this convex hull is actually a convex cone. We will see below the interpretation of the sum in terms of preferences.

Due to our definition of the sum operator, we consider the closed cone: for example, the inputs themselves (e.g. $\vectutil{u}$) fit in our definition, and so is the total indifference $\vectutil{\pi}(\vecteur{0})$.
That would not be the case if we took only $\vectutil{\pi}(a \vecteur{u} + a' \vecteur{v} + b \vecteur{1})$, where $a > 0$ and $a' > 0$.
The purpose of this convention is to have a concise wording for Theorem~\ref{thm:unanimite} that follows.

We now prove that, if $A$ is the set of the utility vectors of a population, then $\sum A$ is the set of utility vectors that respect the unanimous preferences of the population.

\begin{thm}[characterization of the sum]\label{thm:unanimite}
Let $A \in \mathcal{P}(\ens{U}_m)$ and $\vectutil{v} \in \ens{U}_m$.

The following conditions are equivalent.
\begin{enumerate}
	\item $\vectutil{v} \in \sum A$.
	\item $\forall (L,M) \in {\ens{L}_m}^2$:
$
\big[\left(
\forall \vectutil{u} \in A, L \moinssim_{\vectutil{u}} M
\right)
\Rightarrow L \moinssim_{\vectutil{v}} M\big].
$
\end{enumerate}
\end{thm}

\begin{proof}
First, let us remark that the tangent polytope $\mathcal{T}$ generates the tangent hyperplane $\ens{H}$ by positive multiplication. That is: 
$$
\forall \vecteur{\Delta} \in \ens{H}, \exists \vecteur{LM} \in \ens{T}, \exists \lambda \in \mathopen(0,+\infty\mathclose)
\tq \vecteur{\Delta} = \lambda \vecteur{LM}.
$$
Indeed, $\mathcal{T}$ contains a neighborhood of the origin in vector space $\ens{H}$.

Let $\vecteur{v} \in \vectutil{\pi}^{-1}(\vectutil{v})$. We have the following equivalences.
\begin{itemize}
	\item $\forall (L,M) \in {\ens{L}_m}^2, \left( \forall \vectutil{u} \in A, L \moinssim_{\vectutil{u}} M \right) \Rightarrow L \moinssim_{\vectutil{v}} M$,
	\item $\forall \vecteur{LM} \in \ens{T}, \left( \forall \vecteur{u} \in \vectutil{\pi}^{-1}(A), \bracket{\vecteur{u}}{\vecteur{LM}} \geq 0 \right) \Rightarrow \bracket{\vecteur{v}}{\vecteur{LM}} \geq 0$,
	\item $\forall \vecteur{\Delta} \in \ens{H}, \left( \forall \vecteur{u} \in \vectutil{\pi}^{-1}(A), \bracket{\vecteur{u}}{\vecteur{\Delta}} \geq 0 \right) \Rightarrow \bracket{\vecteur{v}}{\vecteur{\Delta}} \geq 0$ (because $\ens{T}$ generates $\ens{H}$),
	\item $\bigcap_{\vecteur{u} \in \vectutil{\pi}^{-1}(A)} \vecteur{u}^{+} \subset \vecteur{v}^{+}$,
	\item $\vecteur{v}$ is in the convex cone of $\vectutil{\pi}^{-1}(A)$ (because of the duality seen in Section~\ref{sec:dualite}),
	\item $\vectutil{v} \in \sum A$.
\end{itemize}\end{proof}

\begin{exemple}
Let us consider a non-indifferent $\vectutil{u}$ and let us examine the sum of $\vectutil{u}$ and its additive inverse $-\vectutil{u}$.
By a direct application of the definition, the sum consists of $\vectutil{u}$, $-\vectutil{u}$ and~$\vectutil{0}$.

However, we have just proved that the sum is the subset of utility vectors preserving the unanimous preferences over lotteries. Intuitively, we could think that, since $\vectutil{u}$ and $-\vectutil{u}$ seem to always disagree, any utility vector $\vectutil{v}$ respects the empty set of their common preferences; so, their sum should be the whole space. But this is not a correct intuition.

Indeed, let us examine the example of $\vecteur{u} = (1,0,\ldots,0)$.
About any two lotteries $L$ and $M$, inverse opinions $\vectutil{u}$ and $-\vectutil{u}$ agree iff $L_1 = M_1$: in that case, both $\vectutil{u}$ and $-\vectutil{u}$ are indifferent between $L$ and $M$.
The only points of the utility space meeting this common property are $\vectutil{u}$ and $-\vectutil{u}$ themselves, as well as the indifference point~$\vectutil{0}$.
\end{exemple}

\section{Riemannian Representation of the Utility Space}\label{sec:riemann}

Since the utility space is a manifold, it is a natural desire to endow it with a Riemannian metric.
In this section, we prove that there is a limited choice of metrics that are coherent with the natural properties of the space and with the \emph{a priori} symmetry between the candidates.

First, let us note that the indifference point $\vectutil{0}$ must be excluded. Indeed, its unique open neighborhood is $\ens{U}_m$ in whole, and no distance is consistent with this property\footnote{Technically, this remark proves that $\ens{U}_m$ (with its natural quotient topology) is not a \emph{$T_1$ space}~\cite{guenard1985complements}.}. In contrast, $\ens{U}_m \setminus \{\vectutil{0} \}$ has the same topology as a sphere of dimension $m-2$, so it can be endowed with a distance.

Now, let us define the round metric. 
The quotient $\R^m / \linearspan(\vecteur{1})$ is identified to $\ens{H}$ and endowed with the inner product inherited from the canonical one of $\R^m$.
The utility space $\ens{U}_m \setminus \{\vectutil{0} \}$ is identified to the unit sphere of $\ens{H}$ and endowed with the induced Riemannian structure.
We note~$\xi_0$ this Riemannian metric on $\ens{U}_m \setminus \{\vectutil{0} \}$.

In order to get an intuitive vision of this metric, one can represent any position $ \vectutil{u} $ by a vector $ \vecteur{u} $ that verifies $\sum u_i = 0$ and $\sum u_i^2 = 1$.
We obtain the $(m-2)$-dimensional unit sphere of $\ens{H}$ and we consider the metric induced by the canonical Euclidean metric of $\R^m$. That is, distances are measured on the surface of the sphere, using the restriction of the canonical inner product on each tangent space. For $m = 3$, such a representation has already been illustrated in Figure~\ref{fig:circle}.

With this representation in mind, we can give a formula to compute distances with~$\xi_0$.
We denote $J$ the $m \times m$ matrix whose all coefficients are equal to 1 and
$P_\ens{H}$ the matrix of the orthogonal projection on~$\mathcal{H}$:
$$
P_\ens{H} = \text{Id} - \dfrac{1}{m} J.
$$
The canonical Euclidean norm of $\vecteur{u}$ is denoted $\| \vecteur{u} \|$.
For two non-indifferent utility vectors $\vecteur{u}$ and $\vecteur{v}$, the distance between $\vectutil{u}$ and $\vectutil{v}$ in the sense of metric $\xi_0$ is:
$$
d(\vectutil{u},\vectutil{v}) = \arccos \bigbracket{%
\dfrac{P_\ens{H} \vecteur{u}}{\| P_\ens{H} \vecteur{u} \|} 
}{%
\dfrac{P_\ens{H} \vecteur{v}}{\| P_\ens{H} \vecteur{v} \|}
}.
$$
If $\vecteur{u}$ and $\vecteur{v}$ are already unit vectors of $\ens{H}$, i.e. canonical representatives of their equivalence classes $\vectutil{u}$ and $\vectutil{v}$, then the formula amounts to $d(\vectutil{u}, \vectutil{v}) = \arccos\bracket{\vecteur{u}}{\vecteur{v}}$. 

A natural property for the distance would be that its geodesics coincide with the unanimity segments defined by the sum. Indeed, imagine that Betty with utilities $\vectutil{u}_B$ is succesfully trying to convince Alice with initial utilities $\vectutil{u}_A$ to change her mind. Assume that Alice evolves continuously, so that her utility follows a path from $\vectutil{u}_A$ to $\vectutil{u}_B$. Unless Betty argues in a quite convoluted way, it makes sense to assume on one hand that Alice's path is a shortest path, and on the other hand that whenever Alice and Betty initially agree on the comparison of two lotteries, this agreement continues to hold all along the path. This is precisely what assumption \ref{enumi:geodesics} below means: shortest paths preserve agreements that hold at both their endpoints.

We will now prove that for $m\geq 4$, the spherical representation is the only one that is coherent with the natural properties of the space and that respects the \emph{a priori} symmetry between candidates.

\begin{thm}[Riemannian representation of the utility space]\label{thm_sphere}
We assume that $m \geq 4$. Let $\xi$ be a Riemannian metric on $\ens{U}_m \setminus \{\vectutil{0} \}$.

Conditions \ref{enumi:hypotheses_spheres} and \ref{enumi:sphere} are equivalent.
\begin{enumerate}
\item\label{enumi:hypotheses_spheres}
\begin{enumerate}
	\item\label{enumi:geodesics} For any non-antipodal pair of points $\vectutil{u},\vectutil{v} \in \ens{U}_m \setminus \{\vectutil{0} \}$ (i.e. $\vectutil{v}\neq -\vectutil{u}$), the set
	$\sum\{\vectutil{u},\vectutil{v}\}$ of elements respecting the unanimous preferences of $\vectutil{u}$ and $\vectutil{v}$ is a segment of a geodesic of $\xi$; and
	\item\label{enumi:permutations} for any permutation $\sigma$ of $\intent{1}{m}$, the action $\Phi_\sigma$ induced on $\ens{U}_m \setminus \{\vectutil{0} \}$ by
$$
(u_1,\ldots,u_m) \to (u_{\sigma(1)},\ldots,u_{\sigma(m)})
$$
is an isometry.
\end{enumerate}
\item\label{enumi:sphere} $\exists \lambda \in \mathopen(0,+\infty\mathclose) \tq \xi = \lambda \xi_0$.
\end{enumerate}
\end{thm}

\begin{proof}
Since the implication \ref{enumi:sphere} $\Rightarrow$ \ref{enumi:hypotheses_spheres} is obvious, we now prove \ref{enumi:hypotheses_spheres} $\Rightarrow$ \ref{enumi:sphere}. The deep result behind this is a classical theorem of Beltrami, which dates back to the middle of the nineteenth century: see \cite{beltrami1866resolution} and \cite{beltrami1869essai}.

The image of a 2-dimensional subspace of $\ens{H}$ in $\ens{U}_m \setminus \{\vectutil{0} \}$  by the canonical projection $\vectutil{\pi}$ is called a \emph{line} in $\ens{U}_m \setminus \{\vectutil{0} \}$. This notion is deeply connected to the summation operator: indeed, the sum of two non antipodal points in $\ens{U}_m \setminus \{\vectutil{0} \}$ is a segment of the line joining them. Condition \ref{enumi:geodesics} precisely means that the geodesics of $\xi$ are the lines of 
$\ens{U}_m\setminus\{\vectutil{0}\}$.
Beltrami's theorem then states that
$\ens{U}_m \setminus \{\vectutil{0} \}$ has constant curvature. Note that this result is in fact more subtle in dimension $2$ (that is, for $m=4$) than in higher dimensions; see \cite{spivak1979vol3}, Theorem~1.18 and \cite{spivak1979vol4}, Theorem~7.2
for proofs.

Since $\ens{U}_m \setminus \{\vectutil{0} \}$ is a topological sphere, this constant curvature must be positive. Up to multiplying $\xi$ by a constant, we can assume that this constant curvature is $1$.
As a consequence, there is an isometry $\Psi:\ens{S}_{m-2}\to \ens{U}_m \setminus \{\vectutil{0} \}$, where $\ens{S}_{m-2}$ is the unit sphere of $\R^{m-1}$, endowed with its usual round metric.
The function $\Psi$ obviously maps geodesic to geodesics, let us 
deduce the following.
\begin{lem}
There is a linear map $\Lambda:\mathbb{R}^{m-1}\to \ens{H}$
inducing $\Psi$, that is such that:
$$\Psi\circ \Pi = \Pi \circ \Lambda,$$
where $\Pi$ denotes both projections $\R^{m-1}\to \ens{S}_{m-2}$ and
$\ens{H}\to \ens{U}_m \setminus \{\vectutil{0} \}$.%
\end{lem}

\begin{proof}
First, $\Psi$ maps any pair of antipodal points of $\ens{S}_{m-2}$ to 
a pair of antipodal points of $\ens{U}_m\setminus \{\vectutil{0}\}$: indeed in both cases antipodal pairs are characterized by the fact that there are more than one geodesic containing them. It follows that
$\Psi$ induces a map $\Psi'$ from the projective space
$\mathrm{P}(\mathbb{R}^{m-1})$ (which is the set of lines through the
origin in $\mathbb{R}^{m-1}$, identified with the set of pairs of
antipodal points of $\ens{S}_{m-2}$) to the projective space
$\mathrm{P}(\ens{H})$ (which is the set of lines through
the origin  in $\ens{H}$, identified with the set of pairs of
antipodal points of $\ens{U}_m\setminus \{\vectutil{0}\}$).

The fact that $\Psi$ sends geodesics of $\ens{S}_{m-2}$ to
geodesics of $\ens{U}_m\setminus \{\vectutil{0}\}$ and condition
\ref{enumi:geodesics} together imply that $\Psi'$ sends projective lines
to projective lines. 

As is well known, a one-to-one map $\mathbb{R}^n\to\mathbb{R}^n$
which sends lines to lines must be an affine map; a 
similar result holds in projective geometry, concluding that
$\Psi'$ must be a projective map.
See e.g. \cite{audin} for both results.

That $\Psi'$ is projective exactly means that 
there is a linear map
$\Lambda:\mathbb{R}^{m-1}\to \ens{H}$ inducing $\Psi'$,
which then induces 
$\Psi:\ens{S}_{m-2} \to\ens{U}_m\setminus\{\vectutil{0}\}$.
\end{proof}

Using $\Lambda$ to push the canonical inner product of $\R^{m-1}$, we get that there exists an inner product $(\vecteur{u},\vecteur{v}) \to \phi(\vecteur{u},\vecteur{v})$ on $\ens{H}$ that induces $\xi$, in the sense that $\xi$ is the Riemannian metric obtained by identifying $\ens{U}_m \setminus \{\vectutil{0} \}$ with the unit sphere defined in $\ens{H}$ by $\phi$ and restricting $\phi$ to it.

The last thing to prove is that $\phi$ is the inner product coming from the canonical one on $\R^m$. Note that hypothesis \ref{enumi:permutations} is mandatory, since any inner product on $\ens{H}$ does induce on $\ens{U}_m \setminus \{\vectutil{0} \}$ a Riemannian metric satisfying \ref{enumi:geodesics}.

Each canonical basis vector $\vecteur{e_j}=(0,\dots,1,\dots,0)$ defines a point in $\ens{U}_m \setminus \{\vectutil{0} \}$ and a half-line $\ell_j=\mathbb{R}_+ e_j$ in $\ens{H}$.
Condition \ref{enumi:permutations} ensures that these half-lines are permuted by some isometries of $(\ens{H},\phi)$.
In particular, there are vectors $\vecteur{u_j}\in\ell_j$ that have constant pairwise distance
(for $\phi$).
\begin{lem}
The family $\vecteur{u_1},\dots,\vecteur{u_{m-1}}$ is, up to
multiplication by a scalar, the unique basis of $\ens{H}$ such that
$\vecteur{u_j}\in\ell_j$ and $\sum_{j<m} \vecteur{u_j}\in -\ell_m$.
\end{lem}

\begin{proof}
First, by definition of the $\vecteur{u_j}$, these vectors form a regular simplex and $\sum_j \vecteur{u_j}=\vecteur{0}$. It follows that 
$\vecteur{u_1},\dots,\vecteur{u_{m-1}}$ has the required property and
we have left to show uniqueness.

Assume $\vecteur{v_1},\dots,\vecteur{v_{m-1}}$ is a basis of
$\ens{H}$ such that 
$\vecteur{v_j}\in\ell_j$ and $\sum_{j<m} \vecteur{v_j}\in -\ell_m$. Then 
there are positive scalars $a_1,\dots, a_m$ such that
$\vecteur{v_j} = a_j \vecteur{u_j}$ for all $j<m$, and
$\sum_{j<m} \vecteur{v_j} = a_m \sum_{j<m} \vecteur{u_j}$.

Then $\sum_{j<m} a_j \vecteur{u_j} = \sum_{j<m} a_m \vecteur{u_j}$,
and since the $u_j$ form a basis, it must hold $a_j=a_m$ for all $j$.
\end{proof}

Now consider the canonical inner product $\phi_0$ on $\ens{H}$ that comes from the canonical one on 
$\R^m$.
Since permutations of coordinates are isometries, we get that the vectors $\vecteur{v_j}=\Pi(\vecteur{e_j})$ (where $\Pi$ is now the orthogonal projection from $\mathbb{R}^m$ to $\mathcal{H}$) form a regular simplex for $\phi_0$, so that $\sum_j \vecteur{v_j}=\vecteur{0}$.
It follows that $\vecteur{u_j}=\lambda \vecteur{v_j}$ for some $\lambda>0$ and all $j$.
We deduce that the $\vecteur{u_j}$ form a regular simplex for both $\phi$ and $\phi_0$, which must therefore be multiple from each other.
This finishes the proof of Theorem~\ref{thm_sphere}.
\end{proof}

However, the implication \ref{enumi:hypotheses_spheres} $\Rightarrow$ \ref{enumi:sphere} in the theorem is not true for $m=3$.
For each non-indifferent utility vector, let us consider its representative verifying $\min(u_i) = 0$ and $\max(u_i) = 1$.
This way, $\ens{U}_m \setminus \{\vectutil{0} \}$ is identified to edges of the unit cube in $\R^3$, as in Figure~\ref{fig:cubic}.
We use this identification to endow $\ens{U}_m \setminus \{\vectutil{0} \}$ with the metric induced on these edges by the canonical inner product on $\R^3$.
Then conditions \ref{enumi:geodesics} and \ref{enumi:permutations} of the theorem are met, but not condition \ref{enumi:sphere}.

\medskip
Another remark is of paramount importance about this theorem. Since we have a canonical representative $\vecteur{u}$ for each equivalence class $\vectutil{u}$ as a unit vector of $\ens{H}$, we could be tempted to use it to compare utilities between several agents.

We stress on the fact that this representation cannot be used for interpersonal comparison of utility levels or utility differences.

For example, for two agents, consider the following representatives:
$$
\left\{
\begin{array}{l}
\vecteur{u} = \left( 0.00 ,  0.71 , -0.71 \right), \\
\vecteur{v} = \left( 0.57 , 0.22 , -0.79 \right).
\end{array}
\right.
$$

The fact that $v_3 < u_3$ does not mean that an agent with preferences $\vectutil{v}$ dislikes candidate~3 more than an agent with preferences $\vectutil{u}$.
Similarly, when changing from candidate 1 to candidate 2, the gain for agent $\vectutil{u}$ ($+0.71$) cannot be compared to the loss for agent $\vectutil{v}$ ($-0.35$).

Theorem~\ref{thm_sphere} conveys no message for interpersonal comparison of utilities. Indeed, utilities belonging to two agents are essentially incomparable in the absence of additional assumptions \cite{hammond1991interpersonal}. Taking canonical representatives on the $(m-2)$-dimensional sphere is only used to compute distances between two points in the utility space.

\section{Application: Probability Measures on the Utility Space}\label{sec:pratique}

Once the space is endowed with a metric, it is endowed with a natural probability measure: the uniform measure in the sense of this metric (this is possible because the space has a finite total measure).
We will denote $\mu_0$ this measure, which is thus the
normalized Riemannian volume defined by the metric $\xi_0$.

In practice, to draw vectors according to a uniform probability law over $\ens{U}_m \setminus \{\vectutil{0} \}$, it is sufficient to use a uniform law on the unit sphere in~$\mathcal{H}$. In other words, once one identifies
$\ens{U}_m \setminus \{\vectutil{0} \}$ with the unit sphere in
$\ens{H}$, then $\mu_0$ is exactly the usual uniform measure.

In the present case, the fact that the round sphere has many symmetries
implies additional nice qualities for $\mu_0$, which we sum up in a proposition.
\begin{prop}
Let $\mu$ be any probability measure on 
$\ens{U}_m\setminus \{\vectutil{0}\}$.
\begin{enumerate}
\item Assume that for all $\delta>0$, $\mu$ gives the same probability to
  all the balls in $\ens{U}_m\setminus\{\vectutil{0}\}$ of radius $\delta$
 (in the metric $\xi_0$); then $\mu=\mu_0$.
\item If $\mu$ is invariant under all isometries for the metric $\xi_0$,
  then $\mu=\mu_0$.
\end{enumerate}
\end{prop}

Both characterizations are classical; the first one is (a strengthening
of) the definition of the Riemannian volume, and the second one follows 
from the first and the fact that any two points on the round sphere can
be mapped one to the other by an isometry.

\begin{figure}
\subfigure[Uniform (Impartial Culture).]{
\begin{tikzpicture}[scale=1]
\node at (0,0) {\includegraphics[width=4cm]{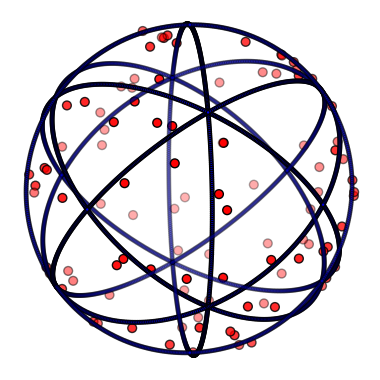}};
\node[shape=circle,fill=black,inner sep=1pt] (A) at (0.2,-1.23) {};
\draw (A) -- (2.2,-1.23) node[right] {Vertex $1 > 2, 3, 4$};
\node[shape=circle,fill=black,inner sep=1pt] (E) at (1.41,-0.16) {};
\draw (E) -- (2.2,-0.16) node[right] {Vertex $1, 4 > 2, 3$};
\draw[fill=black,opacity=0.4] (0.2,0.82) ..controls (0.25,-.2) .. (0.2,-1.23) .. controls (0.55,-1.08) and (0.95,-0.87) .. (1.41,-0.16) .. controls (1.17,.19) and (0.86,.46) .. (0.2,0.82) -- cycle;
\draw (0.9,-0.2) |- (2.2,0.80) node[right] {Facet $1 > 4 > 3 > 2$};
(0.2,0.82) -- cycle;
\draw (1,-0.69) -- (2.2,-0.69) node[right] {Edge $1 > 4 > 2, 3$};
\end{tikzpicture}\label{fig:spherical_culture}
}
\hspace{0.cm}
\subfigure[Von Mises--Fisher.]{
\begin{tikzpicture}
\node at (0,0) {\includegraphics[width=4cm]{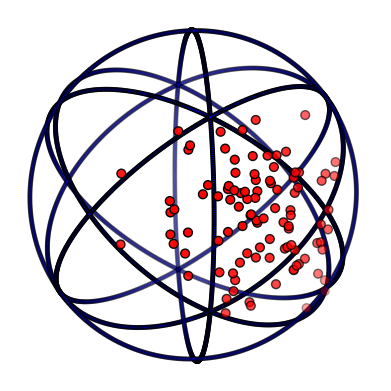}};
\end{tikzpicture}\label{fig:vmf}
}
\caption{Two distributions of 100~agents on $\ens{U}_4$.}
\end{figure}
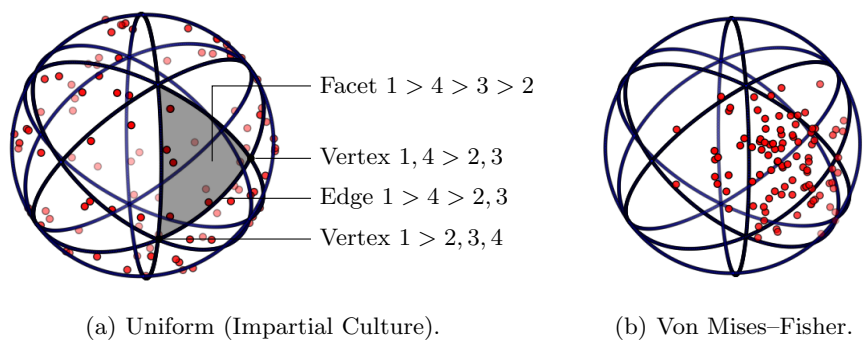

In Figure~\ref{fig:spherical_culture}, we represent a distribution with 100~agents drawn uniformly and independently on the sphere, with 4 candidates. Like for Figure~\ref{fig:sum_vectors}, we represented only the unit sphere of $\mathcal{H}$.

The solid dark lines draw the \emph{permutohedron}, a geometrical figure representing the ordinal aspect of these preferences. Each \emph{facet} is constituted of all the points who share the same strict order of preference. A utility vector belongs to an \emph{edge} if it has only three distinct utilities: for example, if the agent prefers candidate 1 to 4, 4 to 2 and 3, but is indifferent between candidates 2 and 3. Finally, a point is a \emph{vertex} if it has only two distinct utilities: for example, if the agent prefers candidate 1 to all the others, but is indifferent between the others.

In this distribution, each agent has almost surely a strict order of preference. Each order has the same probability and agents are independent, hence this distribution is a natural generalization of the Impartial Culture when considering expected utilities.

Since the point $\vectutil{0}$ is a geometrical singularity, it is difficult to include it naturally in such a measure. If one wants to take it into account, the easiest way is to draw it with a given probability and to use the distribution on $\mathcal{U}_m \setminus \{ \vectutil{0} \}$ in the other cases.
That being said, we have just noticed that all other non-strict orders have a measure equal to 0
; so for a canonical theoretical model, it can be argued that a natural choice is to attribute a measure 0 to the indifference point also. 

Having a distance, hence a uniform measure, allows also to define other measures by their densities with respect to the uniform measure.
Here is an example of a law defined by its density.
Given a vector $\vecteur{u_0}$ in the unit sphere of $\mathcal{H}$ and $\kappa$ a nonnegative real number, the Von Mises--Fisher (VMF) distribution of pole $\vecteur{u_0}$ and concentration $\kappa$ is defined by the following density with respect to the uniform law on the unit sphere in $\mathcal{H}$:
$$
p(\vecteur{u}) = C_\kappa e^{\kappa \bracket{\vecteur{u}}{\vecteur{u_0}}},
$$
where $C_\kappa$ is a normalization constant.
Given the mean resultant vector of a distribution over the sphere, VMF distribution maximizes the entropy, in the same way that, in the Euclidean space, Gaussian distribution maximizes the entropy among laws with given mean and standard deviation.
Hence, without additional information, it is the ``natural'' distribution that should be used.
Figure~\ref{fig:vmf} represents such a distribution, with the same conventions as Figure~\ref{fig:spherical_culture}. To draw a VMF distribution, we used Ulrich's algorithm revised by Wood \cite{ulrich1984computer,wood1994simulation}.

Qualitatively, VMF model is similar to Mallows model, which is used for ordinal preferences \cite{mallows1957nonnull}. In the later, the probability of an order of preference $\sigma$ is:
$$
p(\sigma) = D_\kappa e^{- \kappa \delta(\sigma, \sigma_0)},
$$
where $\sigma_0$ is the mode of the distribution, $\delta(\sigma,\sigma_0)$ a distance between $\sigma$ and $\sigma_0$ (typically Kendall's tau distance), $\kappa$ a nonnegative real number (concentration) and $D_k$ a normalization constant. Both VMF and Mallows models describe a culture where the population is polarized, i.e. scattered around a central point, with more or less concentration.

However, there are several differences.
\begin{itemize}
\item VMF distribution allows to draw a specific point on the utility sphere, whereas Mallows' chooses only a facet of the permutohedron.
\item In particular, the pole of a VMF distribution can be anywhere in this continuum. For example, even if its pole is in the facet $1 > 4 > 3 > 2$, it can be closer to the facet $1 > 4 > \mathbf{2 > 3}$ than to the facet $\mathbf{4 > 1} > 3 > 2$. Such a nuance is not possible in Mallows model.
\item In the neighborhood of the pole, VMF probability decreases like the exponential of the square of the distance (because the inner product is the cosine of the distance), whereas Mallows probability decreases like the exponential of the distance (not squared).
\item VMF is the maximum entropy distribution, given a spherical mean and dispersion (similarly to a Gaussian distribution in a Euclidean space), whereas Mallows' model is not characterized by such a property of maximum entropy.
\end{itemize}

The existence of a canonical measure allows to define other probability measures easily in addition to the two we have just described. Such measures can generate artificial populations of agents for simulation purposes. They can also be used to fit the data from real-life experiments to a theoretical model, and as a neutral comparison point for such data.

To elaborate on this last point, let us stress that given
a (reasonably regular) distribution $\mu$ on a space such as
$\ens{U}_m\setminus\{\vectutil{0}\}$ there is \emph{a priori}
no way to define what it means for an element $\vectutil{u}$
to be more probable than another one $\vectutil{v}$. Indeed,
both have probability $0$ and what would make sense is
to compare the probability of being close to $\vectutil{u}$
to the probability of being close to $\vectutil{v}$. But 
one should compare neighborhoods of the same size, and one 
needs a metric to make this comparison. Alternatively, if one
has a reference distribution such as $\mu_0$, then it makes
sense to consider the density $f=\frac{d\mu}{d\mu_0}$,
which is a (continuous, say) function on $\ens{U}_m\setminus\{\vectutil{0}\}$. Then we can compare
$f(\vectutil{u})$ and $f(\vectutil{v})$ to say whether
one of these elements is more probable than the other according to
$\mu$. Note that in the present case,
comparing the probability of $\delta$-neighborhoods
for the metric $\xi_0$ or the density relative to $\mu_0$
gives the same result in the limit $\delta\to 0$, which is
the very definition of the Riemannian volume.

\section{Conclusion}\label{sec:conclusion}

We have studied the geometrical properties of the classical model of expected utilities, introduced by Von Neumann and Morgenstern, when candidates are considered symmetrical \emph{a priori}.
We have remarked that the utility space may be seen as a dual of the space of lotteries, that inversion and summation operators inherited from $\R^m$ have a natural interpretation in terms of preferences and that the space has a spherical topology when the indifference point is removed.

We have proved that the only Riemannian representation whose geodesics coincide with the projective lines naturally defined by the summation operator and which respects the symmetry between candidates is a round sphere.

All these considerations lay on the principle to add as little information as possible in the system, especially by respecting the \emph{a priori} symmetry between candidates.
This does not imply that the spherical representation of the utility space $\mathcal{U}_m$ is the most relevant one in order to study a specific situation.
Indeed, as soon as one has additional information (for example, a model that places candidates in a political spectrum), it is natural to include it in the model.
However, if one wishes, for example, to study a voting system in all generality, without focusing on its application in a specific field, it looks natural to consider a utility space with a metric as neutral as possible, like the one defined in this paper by the spherical representation.

\noindent\textbf{Acknowledgements.} The work presented in this paper has been partially carried out at LINCS (\url{http://www.lincs.fr}).

\end{document}